\documentclass[prl,reprint,twocolumn,aps,english,superscriptaddress,floatfix,longbibliography]{revtex4-2}
\usepackage[dvipdfmx]{graphics,graphicx}	%タイプセット用
\usepackage{here}
\usepackage{amsthm}
\usepackage{arydshln}
\usepackage{setspace}
\usepackage{wrapfig}
\usepackage{array,amsmath, amssymb, latexsym, bm, mathtools, braket, multirow, url}
\usepackage{color}
\usepackage{cancel}
\usepackage{comment}
\usepackage{ascmac}
\usepackage{tikz}
\usepackage{paralist}
\usepackage{hyperref}

\usepackage[all]{xy}
\usepackage{xr}

\usepackage{bbm}

\theoremstyle{plain}
\newtheorem{theorem}{Theorem}
\newtheorem{definition}[theorem]{Definition}
\newtheorem{lemma}[theorem]{Lemma}

\theoremstyle{definition}
\newtheorem{remark}[theorem]{Remark}
\newtheorem{example}[theorem]{Example}

\newcommand{\cB}{\mathcal{B}}

\newcommand{\cD}{\mathcal{D}}
\newcommand{\cE}{\mathcal{E}}

\newcommand{\cH}{\mathcal{H}}

\newcommand{\cK}{\mathcal{K}}
\newcommand{\cL}{\mathcal{L}}
\newcommand{\cM}{\mathcal{M}}
\newcommand{\cN}{\mathcal{N}}

\newcommand{\cS}{\mathcal{S}}

\newcommand{\cV}{\mathcal{V}}

\DeclareMathOperator{\Tr}{Tr}
\DeclareMathOperator{\id}{id}

\newcommand{\ketbra}[2]{\ket{#1}\!\bra{#2}}

\newcommand{\Her}[1]{\cL_{\mathrm{H}}(#1)}
\newcommand{\Psd}[1]{\cL_{\mathrm{H}}^+(#1)}

\newcommand{\PD}[1]{\cD_{\mathrm{P}}(#1)}
\newcommand{\QD}[1]{\cD_{\mathrm{Q}}(#1)}
\newcommand{\ID}{\cD_{\mathrm{I}}}
\newcommand{\spn}[1]{\mathrm{Span}(#1)}

\begin{document}

\title{Quantum Implementation of Non-Positive-Operator-Valued Measurements in General Probabilistic Theories by Post-Selected POVMs}

\author{Hayato Arai}
\email{hayato.arai@riken.jp}
\affiliation{Mathematical Quantum Information RIKEN Hakubi Research Team, 
RIKEN Cluster for Pioneering Research (CPR) and RIKEN Center for Quantum Computing (RQC), 
Wako, Saitama 351-0198, Japan.
}

\author{Masahito Hayashi}
%\email{hayashi@sustech.edu.cn}
\email{hmasahito@cuhk.edu.cn}
\affiliation{School of Data Science, The Chinese University of Hong Kong, Shenzhen, Longgang District, Shenzhen, 518172, China}
%\email{masahito@math.nagoya-u.ac.jp}
\affiliation{International Quantum Academy, Shenzhen 518048, China}
\affiliation{Graduate School of Mathematics, Nagoya University, Furo-cho, Chikusa-ku, Nagoya, 464-8602, Japan}

\begin{abstract}
It is important problem to clarify the class of implementable quantum measurements from both fundamental and applicable viewpoints.
Positive-Operator-Valued Measure (POVM) measurements are implementable by the indirect measurement methods, and the class is the largest class determined by the mathematical structure of Hilbert space.
However, if we assume probabilistic consistency in our operations instead of the structure of Hilbert space, we can deal with Non-Positive-Operator-Valued Measure (N-POVM) measurements in the framework of General Probabilistic Theories (GPTs).
N-POVM measurements are not considered as implementable, but this paper gives a constructive way to implement N-POVM measurements by POVM measurements and post-selection in quantum theory when we restrict the domain of target states.
Besides, we show that a post-selected POVM measurement is regarded as an N-POVM measurement in a restricted domain.
These results provide a new relationship between N-POVM measurements in GPTs and post-selection.
\end{abstract}

\maketitle

\section{Introduction}
\label{section:introduction}

In quantum theory, studies of quantum measurements are important from both fundamental and applicable viewpoints \cite{Neumann1953,Kraus1969,Ozawa1984,AAV1988,DSS1989}. The most fundamental class of measurements is projective measurements, but as studies of quantum information theory have been flourishing, it has becomed a standard way to consider Positive-Operator-Valued Measure (POVM) measurements \cite{Peres2002}. POVM measurements are the mathematically most general class from the mathematical structure of Hilbert space, and a POVM measurement is implemented as the indirect measurement method by projective measurements \cite{Ozawa1984}.

%Ozawa,

On the other hand, if we assume probabilistic consistency in our operations instead of the structure of Hilbert space, we can deal with a more general mathematical model of measurements. The mathematically general models are known as General Probabilistic Theories (GPTs) \cite{Janotta2014,Plavala2021,KMI2009,Short2010,Barnum2012,Barrett2007,
Muller2013,KBBM2017,Matsumoto2018,Barnum2019,Takagi2019,ALP2019,
ALP2021,MAB2022,AH2024,CBSS,Arai2019,YAH2020,PR1994,Plavala2017,PNL2023}, which is a modern framework for the foundation of quantum theory. In the framework of GPTs, all measurements are available whenever the measurement outcome is obtained with positive probability. Due to this concept, Non-Positive-Operator-Valued Measure (N-POVM) measurements are available in GPTs if we restrict the class of available states.

It is known that N-POVM measurements have the potential to outperform some information tasks in quantum theory \cite{Arai2019,YAH2020,PR1994,Plavala2017,PNL2023}. For example, N-POVM measurements can discriminate non-orthogonal states perfectly \cite{Arai2019,YAH2020}, which violates the quantum Helstrom bound. Oppositely, when we impose the quantum Helstrom bound in a model where N-POVM measurements can be available, the model must be quantum theory \cite{AH2024}, and as a result, any N-POVM measurement is not available in the model. In this sense, N-POVM measurements are unavailable in quantum theory whenever we cannot outperform the quantum Helstrom bound.

However, even in quantum theory, we can outperform the quantum Helstrom bound if an information theoretical ``cheating " is allowed. A typical example of cheating is post-selection \cite{Aaronson2005,CF2015,Arvidsson2020,Regula2022,RLW2022}, which enables us to reject some outcomes of a measurement process.
It is well-known as ambiguous state discrimination that post-selected POVM measurements can improve the error probability of state discrimination and violate quantum Helstrom bound \cite{Chefles1998,HHH2008,SHHH2009,DFY2009,Jun2014}.
Therefore, it is still unclear whether N-POVM measurements in GPTs can be implemented in quantum theory when we allow post-selection.

Surprisingly, this paper shows the possibility of the implementation of N-POVM measurements with post-selected POVM measurements, and we give a constructive way to implement N-POVM measurements in quantum theory.
Our aim is to find a POVM measurement $\{M_i\}_{i\in I\cup \{i_0\}}$ for an N-POVM measurement $\{N_i\}_{i\in I}$ such that
\begin{align}\label{eq:def-IWP}
		\Tr \rho N_i = \cfrac{\Tr \rho M_i}{\sum_{j\in I}\Tr \rho M_j} \quad (\forall i\in I, \ \forall \rho).
\end{align}
However, while $\Tr (\cdot) N_i$ is always linear, the right-hand side $\cfrac{\Tr (\cdot) M_i}{\sum_{j\in I}\Tr (\cdot) M_j}$ is non-linear in general.
To solve this problem, we restrict the domain of states $\rho$ to a subset $\cD$ of density matrices, which corresponds to a prior information of target states.

As we see in Section~\ref{section:main},
we can always find a POVM measurement and a domain satisfying \eqref{eq:def-IWP} for any N-POVM measurement (Theorem~\ref{theorem:general}).
In other words, we can always find a domain where the right-hand side in \eqref{eq:def-IWP} is linear.
In our implementation, there are two types of costs.
The first cost is the dimension of domain, which corresponds to the amount of prior information.
The second is the accepting probability of post-selection.
This paper also gives an upper bound for these two types of costs (Theorem~\ref{theorem:dimension}).

In contrast to the implementation of N-POVM measurements with post-selection,
we also show that the opposite direction is possible, i.e., post-selected POVM can be implemented by N-POVM in a domain (Theorem~\ref{2theorem:post-selection}).
As a result, this paper gives a new relationship between N-POVM measurements in GPTs and post-selected POVM measurements.
By applying these results, we obtain the bound performance for some information tasks by N-POVM measurements and post-selected POVM measurements from one to another.

Our correspondence between post-selected POVM measurements and N-POVM measurements needs to restricted target states to an implementation domain.
In order to see that actual quantum information processes allow the restriction,
we consider a case of ambiguous state discrimination
in Section~\ref{section:ASD}.
When target states generated by a group action are perfectly distinguishable by a post-selected POVM measurement,
there exists an implementation domain including all target states and corresponding N-POVM measurement.

\section{Setting of GPTs and N-POVM measurements}
\label{section:setting}

As a preliminary, we briefly introduce the framework of General Probabilistic Theories (GPTs).
The framework of GPTs imposes a postulate that
any state $\rho$ in the state space $\cS$ and any measurement $\{M_i\}_{i\in I}$ in the measurement space $\cM$ must satisfy $\mathrm{Pr}(\rho,M_i)\ge0$, where $\mathrm{Pr}(\rho,M_i)$ denotes the probability to obtain the outcome $i$ by applying $\{M_i\}$ to $\rho$.
Besides, the standard framework of GPTs imposes that $\cS$ is compact convex set in a real vector space $\cV$.

Because our aim is an implementation in quantum theory,
this paper only considers the case when $\cV$ is given as the set of Hermitian matrices on finite-dimensional Hilbert space $\cH$,
and we denote it $\Her{\cH}$.
However, we remark that we can embed any model of GPTs with a real vector space $\cV$ into the vector space $\Her{\cH}$ whenever $\cV$ is finite-dimensional.

Quantum Theory corresponds to the case that $\cS$ is given as the set of positive semi-definite matrices $\Psd{\cH}$ with trace 1 and $\cN$ is given as the set of POVM measurements, which is the largest measurement space for the state space $\Psd{\cH}$.
However, when we restrict $\cS$ to a subset of $\Psd{\cH}$, we can choose a larger measurement space than POVM measurements.
We call such a measurement \textit{Mon-Positive-Operator-Valued (N-POVM) measurement} defined as follows.
\begin{definition}[N-POVM measurement]
	We say that a family of Hermitian matrices $\bm{N}:=\{N_i\}_{i\in I}$ is a N-POVM measurement
	if at least one $N_i$ is not positive semi-definite and $\sum_{i\in I} N_i=\mathbbm{1}$ holds, where $\mathbbm{1}$ is the identity matrix on $\cH$.
\end{definition}

Next, we define the positive domain of a measurement.

\begin{definition}[Positive Domain and Quantum Domain]
	For a given measurement $\bm{N}=\{N_i\}_{i\in I}$, we define the posiive domain of $\bm{N}$ $\PD{\bm{N}}$ as
	\begin{align}
		\PD{\bm{N}}:=\{\rho\in\Her{\cH}\mid \Tr \rho N_i\ge0 \ \forall i\in I\}.
	\end{align}
	Also, we define quantum domain $\QD{\bm{N}}$ as the intersection of $\PD{\bm{N}}$ and the set of density matrices, i.e.,
	\begin{align}
		\QD{\bm{N}}:=\{\rho\in\Psd{\cH}\mid \Tr \rho=1,\ \Tr \rho N_i\ge0 \ \forall i\in I\}.
	\end{align}
\end{definition}

If and only if $\PD{\bm{N}}\supset\cS$, the measurement $\bm{N}$ satisfies the postulate of GPTs even though $\bm{N}$ is a N-POVM measurement.
In other words, the measurement $\bm{N}$ is available in the state space $\cS$ of GPTs.
Besides, if a quantum state $\rho$ satisfies $\rho\in\QD{\bm{N}}$,
the measurement $\bm{N}$ is consistent to the postulate even for a state $\rho$.

%In the framework of GPTs, we define a model, which is a tuple of a finite-dimensional real vector space $\cV$ with an inner product $\langle, \rangle$, a positive cone $\cC$, and an element $u\in\cV$.

\section{Implementation of N-POVM by Post-Selected POVM}
\label{section:main}

Now, we consider an implementation of N-POVM by post-selected POVM measurement defined as follows.

\begin{definition}[Implementation with Post-Selected POVM measurements]\label{def:IWP}
	Let $\bm{N}:=\{N_i\}_{i\in I}$ and $\bm{M}:=\{M_j\}_{j\in I\cup \{i_0\}}$ be an N-POVM measurement and a POVM measurement, respectively.
	Also, let $\ID\subset\QD{\bm{N}}$ be a subset of the quantum domain.
	We say that $\bm{N}$ is implemented by $\bm{M}$ with post-selection on the implementation domain $\ID$
	if the following relation holds for any $\rho\in\ID$ and $i\in I$:
	\begin{align}\label{eq:def-IWP}
		\Tr \rho N_i = \cfrac{\Tr \rho M_i}{\sum_{j\in I}\Tr \rho M_j}.
	\end{align}
\end{definition}

We illustrate the concept of our implementation as a protocol in quantum theory (Figure~\ref{figure:implementation}).

\begin{figure}[h]
	\centering
	\includegraphics[width=8cm]{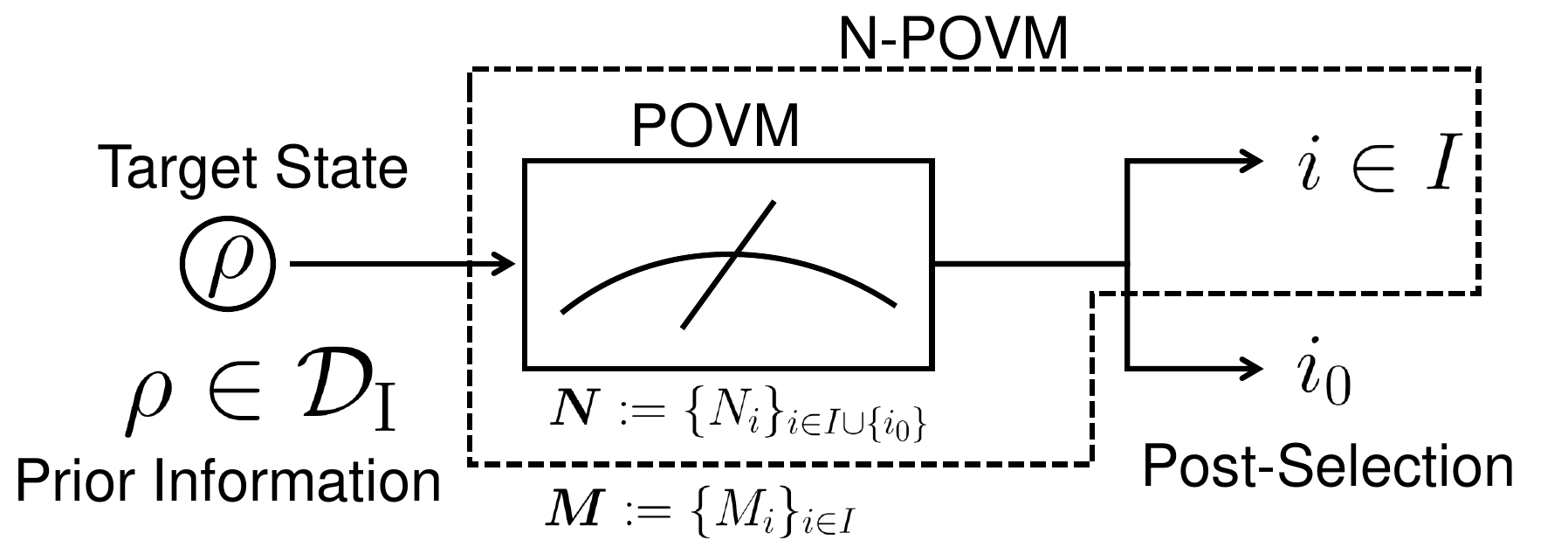}
	\caption{
	Implementation of N-POVM measurement with prior information and post-selection. When we assume prior information of a target state $\rho$, the process of the post-selected POVM measurement is regarded as an N-POVM measurement.
	}
	\label{figure:implementation}
\end{figure}

It is natural to assume convexity on the implementation domain $\ID$,
and therefore, here in after, we assume that $\ID$ is convex.
Then, we have the following lemma.
\begin{lemma}\label{lemma:domain}
Let $\bm{N}:=\{N_i\}_{i\in I}$ 
be an N-POVM measurement.
Assume that for any two elements $\rho_1,\rho_2 \in \ID$
there is an index $i\in I$ such that 
$\Tr N_i \rho_1\neq \Tr N_i \rho_2$. 
When a POVM $\bm{M}:=\{M_j\}_{j\in I\cup \{i_0\}}$ 
implements $\bm{N}$ with post-selection on the implementation domain $\ID$,
the value $\sum_{j\in I}\Tr \rho M_j$ does not depend on $\rho\in\ID$.
\end{lemma}

The proof of Lemma~\ref{lemma:domain} is written in Appendix~\ref{append:lem}.
This lemma guarantees that  
the value $1/\left(\sum_{j\in I}\Tr \rho M_j\right)$ does not depend on 
the choice of $\rho \in \ID$.
Hence, this value is denoted by $\mathrm{Acc}(\bm{M})$.

% and a POVM measurement, respectively.
%	Also, let $\ID\subset\QD{\bm{N}}$ be a subset of the quantum domain.

Next, we consider costs of the implementation.
There exist two types of costs in the above implementation of N-POVM.
The first cost is the value $\mathrm{Acc}(\bm{M})$,
which corresponds to the accepting probability of post selection.
The second cost is the amount of prior information,
which corresponds to the dimension of the implementation domain $\ID$.
Here, 
we denote the value $\dim(\bm{X})$ for a family $\bm{X}$ or a set $\bm{X}$
as the dimension of the linear subspace $\spn{\bm{X}}$, where $\spn{\bm{X}}$ denotes the linear subspace spanned by $\bm{X}$.
Then, we estimate the cost of prior information as $\dim(\ID)$.
We need to know that a target state belongs to $\ID$,
and a large implementation domain $\ID$ corresponds to a small amount of the prior information, and a small implementation domain corresponds to a large amount of the prior information.

We give a constructive way to implement N-POVM as the following theorem.

\begin{theorem}\label{theorem:general}
	Let $\bm{N}:=\{N_i\}_{i\in I}$ be an N-POVM measurement written as
	\begin{align}\label{eq:representation}
		N_i=\sum_{k=1}^n f_i^{(k)}(S_i^{(k)}),
	\end{align}
	where $f_i^{(k)}$ and $S_i^{(k)}$ are a linear map on $\Her{\cH}$ and a positive semi-definite matrix on $\cH$ for any $i,k$, respectively.
	Define a set $\cD(\{f_i^{(k)}\}_{i,k})$ as
	\begin{align}\label{def:Im-domain}
	\begin{aligned}
		\cD(\{f_i^{(k)}\}_{i,k}):=&\{\rho\in\Psd{\cH}\mid \Tr \rho=1,\\
		&f_i^{(k)\dag}(\rho)=\rho \ ( i\in I, \ 1\le k\le n)\}.
	\end{aligned}
	\end{align}
	Then, $\bm{N}$ is implemented by the following POVM $\bm{M}(\{f_i^{(k)},S_i^{(k)}\}):=\{M_j\}_{j\in I\cup\{i_0\}}$ with post-selection on the implementation domain $\ID=\cD(\{f_i^{(k)}\}_{i,k})$:
	\begin{align}
		M_j=
		\begin{cases}
			\displaystyle \frac{1}{c({\{S_i^{(k)}\}})}\sum_{k=1}^n S_i^{(k)} & (j=i\in I),\\
			\displaystyle \mathbbm{1}-\frac{1}{c({\{S_i^{(k)}\}})}\sum_{i\in I}\sum_{k=1}^n S_i^{(k)} & (j=i_0),
		\end{cases}\label{def:Im-POVM}
	\end{align}
	where $c({\{S_i^{(k)}\}})$ is the maximum eigenvalue of the Hermitian matrix $\displaystyle \sum_{i\in I}\sum_{k=1}^n S_i^{(k)}$.
\end{theorem}

The proof of Theorem~\ref{theorem:general} is written in Appendix~\ref{append-1}.
Because any state $\rho\in\cD(\{f_i^{(k)}\}_{i,k})$ satisfies $\sum_{i\in I} \Tr M_i\rho=c(\{S_i^{(k)}\})$, the value $1/c({\{S_i^{(k)}\}})$ equal to the value $\mathrm{Acc}(\bm{M}(\{f_i^{(k)},S_i^{(k)}\}))$.
Also, because of the choice of the $\ID$, the relation $\ID\subset\QD{\bm{N}}$ holds.
Theorem~\ref{theorem:general} states that the right-hand side in \eqref{eq:def-IWP} is linear on $\ID$.

\begin{remark}
	Theorem~\ref{theorem:general} gives a correspondence between two measurements in two models in GPTs through post-selection.
	Actually, Theorem~\ref{theorem:general} is extended to a more general statement for general two classes of measurements.
	We state the extended theorem in Appendix~\ref{append-extend}.
\end{remark}

Theorem~\ref{theorem:general} gives an implementation of N-POVMs with post-selection.
However,
Theorem~\ref{theorem:general} does not state the two costs of implementation.
Especially, Theorem~\ref{theorem:general} does not guarantee that the domain $\ID$ is not the empty set.
Our aim is to find an implementation with large accepting probability and large domain, which depends on the form in \eqref{eq:representation}.
Nevertheless, the definition of N-POVM is quite simple, and therefore, it is difficult to give a convenient sufficient condition for small costs.
Therefore, we give a sufficient condition for small costs in concrete situations as the next two theorems.

The first theorem gives a sufficient condition for a large accepting probability.
%\color{red}
\begin{theorem}\label{theorem:ppt}
	Let $d$ be the dimension of the Hilbert space $\cH$.
	Also, let $\bm{N}=\{N_i\}_{i\in I}$ be an N-POVM given in Theorem~\ref{theorem:general} such that $f_i^{(k)}$ is trace preserving for any $i,k$.
	Then, the POVM $\bm{M}$ given by Theorem~\ref{theorem:general} satisfies
	\begin{align}
		\mathrm{Acc}(\bm{M})\ge 1/d.
	\end{align}
\end{theorem}

Theorem~\ref{theorem:ppt} is an easy consequence of trace preservation of $f_i^{(k)}$,
but
the proof of Theorem~\ref{theorem:ppt} is written in Appendix~\ref{append-3}.
As a typical example,
an N-POVM defined by partial transposition on bipartite Hilbert space satisfies the assumption of Theorem~\ref{theorem:ppt}.
We see a detail of this example in the next section.
%The partial transposition is trace preserving, and therefore, an N-POVM satisfying \eqref{eq:ppt} has accepting probability larger than $1/d$.

Next,
to state the second theorem,
which gives a sufficient condition for a large domain,
we define $\dim'(\bm{N})$ as follows.
Let us consider a Hilbert space $\cH$ with dimension $d$, and we complete orthogonal basis $\{B_i\}_{i=1}^{d^2}$ of $\Her{\cH}$ with $\dim(\cH)=d$ such that $\{B_i\}_{i=1}^d$ composes orthogonal projections on $\cH$.
Then, $\dim'(\bm{N})$ is defined as the smallest number of elements of such basis $\{B_i\}_{i\in I}$ composing all the element of $\bm{N}$.
Also, we define an $\epsilon$-ball of $x\in\Her{\cH}$ as $\cB_\epsilon(x):=\{y\in\Her{\cH}\mid \|x-y\|_2\le \epsilon\}$.

%Here, we give a typical case when the domain $\ID$ is not empty.
%Let us consider all functions $f_k$ are the same linear positive map $P$.
%In other words, $P$ is a linear function from the set of density matrices to itself.
%Because the set of density matrices is compact and convex,
%Brouwer's fixed-point theorem ensures that there is at least one fixed point in the set of density matrices.
%As a result, the domain $\cD$ is not empty.

%\color{red}
\begin{theorem}\label{theorem:dimension}
	Let $\cH$ be a Hilbert space $\cH$ with dimension $d$.
	Also, let $\bm{N}=\{N_i\}_{i\in I}$ be an N-POVM measurement whose quantum domain $\QD{\bm{N}}$ satisfies
	\begin{itemize}
		\item[C1] $\QD{\bm{N}}$ contains $d$-number of orthogonal pure states $\rho_j$.
		\item[C2] $\QD{\bm{N}}$ contains $\cB_\epsilon(\frac{\mathbbm{1}}{d})$ for some $\epsilon>0$.
	\end{itemize}
	%whose element $N_i$ satisfies $\Tr \rho N_i\ge0$ for any $i\in I$ and any separable state $\rho$.	
%whose element $N_i$ is not a positive-semi definite for $i\in I'\subset I$.
Then, there exists a pair of a POVM $\bm{M}=\{M_i\}_{i\in I\cup\{i_0\}}$ 
and an implementation domain $\ID$ that satisfies the following conditions;
(i) The pair implements $\bm{N}$ with post-selection.
(ii) The relations 
%and satisfies a family of linear functions $\{f_i\}_{i\in I}$ and a family of positive-semi definite matrices $\{S_i\}_{i\in I}$ such that
	\begin{gather}
		%N_i=f_i(S_i) \quad (i\in I),\label{eq:dim-1}\\
		\dim(\ID)\ge d^2+d-2\dim'(\{N_i\}_{i\in I'}),\label{eq:dim-2}\\
		\mathrm{Acc}(\bm{M})\ge 1/d\label{eq:dim-3}
	\end{gather}
hold.
	%where $\{f_i\}_{i\in I}$ and $\{S_i\}_{i\in I}$ are the family of linear maps and positive semi-definite matrices determined by Theorem~\ref{theorem:general}.
\end{theorem}

%\color{red}
%\begin{theorem}\label{theorem:dimension}
	%Let $d$ be the dimension of the Hilbert space $\cH$.
%	Let $\bm{N}=\{N_i\}_{i\in I}$ be an N-POVM measurement.
%whose element $N_i$ is not a positive-semi definite for $i\in I'\subset I$.
%	Then, there exist a POVM $\bm{M}=\{M_i\}_{i\in I\cup\{i_0\}}$ and implementation domain $\ID$ that implements $\bm{N}$ with post-selection and satisfies
%a family of linear functions $\{f_i\}_{i\in I}$ and a family of positive-semi definite matrices $\{S_i\}_{i\in I}$ such that
%	\begin{gather}
		%N_i=f_i(S_i) \quad (i\in I),\label{eq:dim-1}\\
%		\dim(\cD(\{f_i\}))=\max\{\QD{\bm{N}}-2\dim(\{N_i\}_{i\in I'}),0\}.\label{eq:dim-2}
		%\mathrm{Acc}(\bm{M}(\{f_i^{(k)},S_i^{(k)}\}))\le d.\label{eq:dim-3}
%	\end{gather}
%\end{theorem}
%\color{black}

The proof of Theorem~\ref{theorem:dimension} is written in Appendix~\ref{append-2}.
Due to Theorem~\ref{theorem:dimension},
we can choose a large implementation domain if the number of outcome is small.
In other words, the number of outcome gives a sufficient condition for a implementation of N-POVM with a small amount of cost about prior information.
As a typical example,
an N-POVM whose elements are positive for all separable states in bipartite Hilbert space satisfies the assumption of Theorem~\ref{theorem:ppt}.
Actually, positivity for separability ensures the condition C1 in Theorem~\ref{theorem:dimension},
and so-called separable ball theorem \cite[Corollary 4]{GH2002} ensures condition C2.
Such an N-POVM is well-known as an available measurement in quantum composite system in GPTs \cite{Janotta2014,Plavala2021,ALP2019,
ALP2021,MAB2022,AH2024,Arai2019,YAH2020}.
Especially, an N-POVM defined by partial transposition satisfies the assumption, similarly to the case of Theorem~\ref{theorem:ppt}.
We see a detail of this example in the next section.
%The partial transposition is trace preserving, and theref

\section{N-POVM Defined by Partial Transposition}
\label{section:PT}

In this section,
we consider a more concrete case as an example,
and we see that our implementation of N-POVMs improve state discrimination with POVMs.

Let us consider a bipartite Hilbert space $\cH=\cH_A\otimes\cH_B$,
and let $\Gamma$ be the partial transposition map transposing on the system $B$.
We consider an N-POVM measurement $\bm{N}=\{N_i\}_{i\in I}$ satisfying
\begin{align}\label{eq:ppt}
	N_i=S_i^{(1)}+\Gamma(S_i^{(2)}),
\end{align}
where $S_k^{(i)}$ is a positive semi-definite matrix.
In other words, we consider the case when $f_i^(1)=\id$ and $f_i^(2)=\Gamma$.
In this case, Theorem~\ref{theorem:general} gives a domain $\cD(\Gamma)=\{\rho\in\Psd{\cH}\mid \Tr \rho=1, \     \Gamma(\rho)=\rho\}$.
The domain $\cD(\Gamma)$ contains the set of all separable states $\sum_i p_i \rho_i^A\otimes\rho_i^B$ such that all entries of $\rho_i^B$ are real numbers for any $i$.
Therefore, $d(\cD(\Gamma))$ is larger than $\frac{1}{2}\dim(\cH_A)^2\left(\dim(\cH_B)^2+\dim(\cH_B)\right)$,
which is especially larger than the half of the full-dimension $d^2=\dim(\cH_A)^2\dim(\cH_B)^2$.
Also, Theorem~\ref{theorem:ppt} ensures the accepting probability is larger than $1/d$.

%Also, similarly to Theorem~\ref{theorem:dimension}, the value $c_{\bm{N}}$ given by Theorem~\ref{theorem:general} satisfies $c_{\bm{N}}\le d$, which gives a lower bound of the accepting probability of post-selection.

An N-POVM given by \eqref{eq:ppt} is available in any model of quantum composite system in GPTs \cite{Janotta2014}.
Especially, in the case $(\dim(\cH_A),\dim(\cH_B))=(2,2),(2,3),(3,2)$, all available measurement consistent with fundamental postulate are written in the form \eqref{eq:ppt}.

Moreover,
the reference \cite{Arai2019} shows that such N-POVM can improve the performance for state discrimination.
We see this in the following concrete example in the case of $(\dim(\cH_A),\dim(\cH_B))=(2,2)$.

\begin{example}
	Let $\bm{N}=\{N_0,N_1\}$ be a family of matrices defined as
	\begin{align}
		N_0:=&
		\begin{bmatrix}
			1&0&0&0\\
			0&0&1&0\\
			0&1&0&0\\
			0&0&0&1
		\end{bmatrix}, \quad
		N_1:=
		\begin{bmatrix}
			0&0&0&0\\
			0&1&-1&0\\
			0&-1&1&0\\
			0&0&0&0
		\end{bmatrix}\\
		=&\Gamma\left(
		\begin{bmatrix}
			1&0&0&1\\
			0&0&0&0\\
			0&0&0&0\\
			1&0&0&1
		\end{bmatrix}
		\right).
	\end{align}
	The family $\bm{N}$ is composed of PPT matrices, and therefore, $\bm{N}$ is a measurement in $\mathrm{SEP}(A;B)$.
	
	Also, we consider another measurement $\bm{M}:=\{M_0,M_1,M_2\}$, where $M_0:=\frac{1}{2}\Gamma(N_0)$, $M_1:=\frac{1}{2}N_1$, and $M_2:=\mathbbm{1}-M_0-M_1$.
	$M_2$ is written as
	\begin{align}
		M_2=\frac{1}{2}
		\begin{bmatrix}
			1&0&0&-1\\
			0&1&1&0\\
			0&1&1&0\\
			-1&0&0&1
		\end{bmatrix},
	\end{align}
	and therefore, $M_2$ is positive semi-definite.
	Hence, the measurement $\bm{M}$ is a POVM.
	
	Then, we show that the measurement $\bm{N}$ is implemented by $\bm{M}$ with post-selection on the domain $\cD$ as follows.
	For any state $\rho\in\cD$,
	the following equation holds:
	\begin{align}\label{eq:ex1-1}
		&\sum_{j=0,1}\Tr \rho M_j
		=1-\Tr \rho M_2\nonumber\\
		=&1-\Tr \rho \frac{1}{2}(M_2+\Gamma(M_2))
		=1-\Tr \rho \frac{1}{2}\mathbbm{1}=\frac{1}{2}.
	\end{align}
	By applying the equation \eqref{eq:ex1-1},
	we obtain the following equation:
	\begin{align}
		&\cfrac{\Tr \rho M_0}{\sum_{j\in I}\Tr \rho M_j}
		=2\Tr \rho M_0\nonumber\\
		=&\Tr \rho\Gamma(N_0)
		=\Tr \Gamma(\rho) N_0
		=\Tr \rho N_0.
	\end{align}
	Also, we obtain the following equation:
		\begin{align}
		&\cfrac{\Tr \rho M_1}{\sum_{j\in I}\Tr \rho M_j}
		=2\Tr \rho M_1
		=\Tr \rho N_1.
	\end{align}
	As a result,
	the measurement $\bm{N}$ is implemented by $\bm{M}$ with post-selection on the domain $\cD(\Gamma)$.
	Here, we remark that the accepting probability is $1/2$, which is larger than $1/d=1/4$.
	
	Finally, the measurement $\bm{N}$ can discriminate the following two non-orthogonal states perfectly, i.e., $\Tr \rho_i N_j=\delta_{ij}$ holds:
	\begin{align}
		\rho_0:=
		\begin{bmatrix}
			1&0&0&0\\
			0&0&0&0\\
			0&0&0&0\\
			0&0&0&0
		\end{bmatrix},
		\rho_1:=\frac{1}{2}
		\begin{bmatrix}
			1&1&1&1\\
			1&1&1&1\\
			1&1&1&1\\
			1&1&1&1
		\end{bmatrix}.
	\end{align}
	Here, we remark that $\rho_i$ satisfies $\rho_i\in \cD(\Gamma)$.
	As a result,
	the N-POVM discriminates two non-orthogonal states perfectly, which is impossible by POVM.
	This fact corresponds to the fact that
	$\rho_0$ and $\rho_1$ are perfectly distinguishable with post-selection by the POVM $\bm{M}$.
\end{example}

\section{Implementation of Post-Selected POVM by N-POVM}
\label{section:PS}

Next, we investigate the opposite direction, i.e., an implementation of post-selected POVM measurements by N-POVM measurements.
As we mentioned in Section~\ref{section:introduction}, a post-selected measurement is a non-linear function but a N-POVM measurement is linear, and therefore, we also need to restrict the domain where the post-selected measurement behaves linear.

Let us consider a POVM measurement $\bm{M}=\{M_i\}_{i\in I\cup \{i_0\}}$ and the situation when we reject the outcome $i_0$.
Our aim is to find a domain where the following function is linear:
\begin{align}
	\frac{\Tr M_i\rho}{\sum_{j\in I}\Tr M_j\rho} \quad \rho\in\cD.
\end{align}
In other words, we need to choose $\cD$ as $\sum_{j\in I}\Tr M_j\rho$ is constant. Because $\sum_{j\in I}\Tr M_j(\cdot)$ is linear, we always choose $\cD$ as non-empty.
For example, if $\sum_{j\in I}\Tr M_j$ have a degenerate eigenvalues, we can choose $\cD$ as the set of states corresponds to the eigenspace of the degenerate eigenvalues.
By restricting the domain $\cD$, we can choose an N-POVM measurement associated with a given post-selected POVM measurement as follows.

\begin{theorem}\label{2theorem:post-selection}
Let $\bm{M}=\{M_i\}_{i\in I\cup \{i_0\}}$ be a POVM measurement.
%Assume that $\ID$ is given as a convex full of at most $d^2$ linearly independent elements $\{\rho_j\}_{j=1}^{k}$ of $\ID$, i.e., $k \le d^2$. 
We consider a linear subspace $\cK\subset \Her{\cH}$ and the intersection $\ID:=\cK\cap\Psd{\cH}$.
%$\{\rho_j\}_{j=1}^{k}$. 
We consider Hilbert Schmidt inner product,
and denote the projection to ${\cal K}$ by ${\cal P}$.
Also, assume that there exists a constant $c_0>0$ such that 
any $\rho\in \ID$ satisfies 
\begin{align}
\Tr \rho M_{i_0}%=\frac{c-1}{c}
=\frac{1}{c_0}
\label{MNB4}.
\end{align}
Then, the following conditions are equivalent.
\begin{description}
\item[(i)]
The following relation holds;
\begin{align}
{\cal P}(c_0 M_{i_0}-\mathbbm{1})=0.\label{MNB1}
\end{align}
\item[(ii)]
There exists an N-POVM $\bm{N}:=\{N_i\}_{i\in I}$
with an implementation domain $\ID$ such that
	\begin{align}\label{2eq:PS-2}
		\Tr N_i\rho=\frac{\Tr M_i\rho}{\sum_{j\in I}\Tr M_j\rho}
	\end{align}
	for any $\rho\in\ID$ and any $i\in I$.
\item[(iii)]
The N-POVM $\bm{N}:=\left\{N_i:=\frac{1}{c_0-1}\left(c_0M_i+\frac{1}{|I|}\left(\mathbbm{1}-c_0M_{i_0}\right)\right)\right\}_{i\in I}$
with an implementation domain $\ID$ satisfies \eqref{2eq:PS-2}.
for any $\rho\in\ID$ and any $i\in I$.
\end{description}
\end{theorem}

The proof of Theorem~\ref{2theorem:post-selection} is written in Appendix~\ref{append-4}.
Due to Theorem~\ref{theorem:general} and Theorem~\ref{2theorem:post-selection},
we give a new relationship between N-POVM measurements and post-selected POVM measurements.

The case of Theorem \ref{theorem:general} satisfies 
the condition (i) of Theorem \ref{2theorem:post-selection}
as follows.
In the case of 
Theorem \ref{theorem:general},
by using $c:=c({\{S_i^{(k)}\}})$,
\eqref{def:Im-POVM} implies
\begin{align}
&\frac{c}{c-1}M_{i_0}-\mathbbm{1}
=\frac{c}{c-1}
\left(\mathbbm{1}-\frac{1}{c}\sum_{i\in I}\sum_{k=1}^n S_i^{(k)}\right)-\mathbbm{1} \notag \\
=&\frac{1}{c-1}
\left(\mathbbm{1}-\sum_{i\in I}\sum_{k=1}^n S_i^{(k)}\right).\label{BN1}
\end{align}
Since $\Tr \rho S_i^{(k)}=\Tr \rho f_i^{(k)}(S_i^{(k)})$ for 
$\rho\in \ID$, we have
\begin{align}
\Tr \rho \sum_{i\in I}\sum_{k=1}^n S_i^{(k)}
=
\Tr \rho \sum_{i\in I}\sum_{k=1}^n f_i^{(k)}(S_i^{(k)})
=1=\Tr \rho \mathbbm{1}.\label{BN2}
\end{align}
Thus, the combination of \eqref{BN1} and \eqref{BN2} guarantees 
the relation 
\begin{align}
&\Tr \rho \left(\frac{c}{c-1}M_{i_0}-\mathbbm{1}\right)=0
\end{align}
for any element $\rho\in \ID$,
which implies the condition (i).

%In fact, the condition $\Tr \rho S_i^{(k)}=\Tr \rho f_i^{(k)}(S_i^{(k)})$ for $\rho\in \ID$ implies ${\cal P}(f_i^{(k)}(S_i^{(k)})-S_i^{(k)})=0$.

\section{Relation to Ambiguous State Discrimination}
\label{section:ASD}
\subsection{General structure}
As seen in Section~\ref{section:PS},
a post-selected POVM measurement can be regarded as an N-POVM measurement without post-selection.
In this section, we see a concrete correspondence between post-selected POVM measurements and N-POVM measurements in the scenario of ambiguous state discrimination.
For our correspondence, the accepting probability for post-selection is constant in an implementation domain, but a special case of ambiguous state discrimination satisfies the condition for the target states.

Let us consider $d$-dimensional Hilbert space ${\cH}$ and
$d$ pure states $\{\ket{\psi_j}\}_{j=1}^d$ such that
$\ket{\psi_1}, \ldots, \ket{\psi_d}$ are linear independent.
We choose its dual basis $\{|\phi_j\rangle\}_{j=1}^d$, which satisfies
\begin{align}
\langle \phi_j|\psi_k\rangle= \delta_{j,k}.
\end{align}
Also, we choose positive numbers $c_j>0$ as
\begin{align}
\sum_{j=1}^d c_j |\phi_j\rangle\langle \phi_j| \le \mathbbm{1}.
\end{align}
Then, we define $M_0:=\mathbbm{1}- \sum_{j=1}^d c_j |\phi_j\rangle\langle \phi_j|$.
The POVM $\{c_j |\phi_j\rangle\langle \phi_j|\}_{j=1}^d\cup\{M_0\}$
discriminates 
the states $\{|\psi_j\rangle\}_{j=1}^d$
perfectly with the ambiguous event $M_0$.
Usually, $c_j$ depends on $j$.
%the probability $\langle \psi_j |M_0|\psi_j\rangle$
As known in \cite{HHHH2010},
the states in $\{f(g)\ket{\psi}\}_{g\in G}$ are perfectly distinguishable
when we choose an optimal input state $\ket{\psi}$.
In this paper, we deal with the case when we do not necessarily choose optimal input states.

We study whether  
Theorem \ref{2theorem:post-selection}
can be applied to the case when
$\ID$ is given as the convex full of $
|\psi_j\rangle\langle \psi_j|$ and the POVM is the above POVM.
When $c_j$ 
%the probability $\langle \psi_j |M_0|\psi_j\rangle$ 
is a constant $c$ for any $j$,
the condition \eqref{MNB1} holds as follows.
For any $j'$, we have
\begin{align}
&\Tr \left(\frac{1}{c}M_0 -\mathbbm{1}\right)|\psi_{j'}\rangle\langle \psi_{j'}|\nonumber\\
=&\Tr \left(\frac{1}{c}\sum_{j=1}^d c |\phi_j\rangle\langle \phi_j| -\mathbbm{1}
\right)|\psi_{j'}\rangle\langle \psi_{j'}|
=\frac{1}{c}c-1=0.
\end{align}
Hence, Theorem \ref{2theorem:post-selection}
guarantees the existence of N-POVM to perfectly distinguish 
the states 
$\{ |\psi_j\rangle\}_{j=1}^d$.

\subsection{Case with group representation}
Next, we seek the case when $c_j$ does not depend on $j$.
Assume that $G$ is a finite group.
Let ${\cal H}_1, \ldots, {\cal H}_k$ be all irreducible representation spaces of $G$.
We have the relation
\begin{align}
\sum_{l=1}^k \dim {\cal H}_l^2=|G|.
\end{align}
We consider the following space
\begin{align}
{\cal H}:=\bigoplus_{l=1}^k {\cal H}_l\otimes {\cal H}_l',
\end{align}
where the Hilbert space ${\cal H}_l'$ has the same dimension 
$d_l$ as ${\cal H}_l$.
In the above, the group $G$ does not act on ${\cal H}_l'$ so that
${\cal H}_l'$ expresses the multiplicity of ${\cal H}_l$.
We denote the presentation of $G$ on ${\cal H}$ by $f$.

We choose a pure state $|\psi\rangle$.
We consider the discrimination of the states
$\{ f(g)|\psi\rangle\}_{g \in G}$.
There uniquely exists a state $|\phi\rangle$ orthogonal to
$\{ f(g)|\psi\rangle\}_{g \in G, g\neq e}$.
Then, the states
$\{ f(g)|\phi\rangle\}_{g \in G}$ form the dual basis of 
$\{ f(g)|\psi\rangle\}_{g \in G}$.
We choose a positive number $c_{|\phi\rangle}>0$ as
\begin{align}
M_0:= I-c_{|\phi\rangle} 
\sum_{g \in G} f(g)|\phi\rangle \langle \phi| f(g)^\dagger\ge0.
\end{align}
Since $M_0$ is invariant, i.e., 
$f(g)M_0f(g)^\dagger=M_0$,
we find that the probability
\begin{align}
 \langle \psi| f(g)^\dagger M_0 f(g)|\psi\rangle
= \langle \psi| M_0 |\psi\rangle
\end{align}
does not depend on $g\in G$.
In other words, the accepting probability is constant for any $f(g)(\ket{\psi})$.
In this way, it is not a strange situation that the accepting probability for post-selection is constant.
Hence, Theorem \ref{2theorem:post-selection}
guarantees the existence of N-POVM to perfectly distinguish 
the states 
$\{ f(g)|\psi\rangle\}_{g \in G}$.

%\subsection{}
To calculate $c_{|\psi\rangle}$ we focus on the structure of 
the state $|\psi\rangle$.
We choose a maximally entangled state $|\phi_l\rangle$ on
${\cal H}_l\otimes {\cal H}_l'$ and define the distribution $p_l$
as $p_l:= \frac{d_l^2}{|G|}$.
Any state $|\psi\rangle$ is written as
\begin{align}
|\psi\rangle = \sum_{l=1}^k \sqrt{p_l} \mathbbm{1}_l\otimes F_l |\phi_l\rangle,
\end{align}
where $F_l$ is an invertible matrix satisfying the following normalization condition:
\begin{align}
\sum_{l=1}^l p_l \Tr F_l F_l^\dagger =1.
\end{align}
Since the states
$\{ f(g)|\phi\rangle\}_{g \in G}$ form the dual basis of 
$\{ f(g)|\psi\rangle\}_{g \in G}$,
we have
\begin{align}
 \langle \phi| f(g)|\psi\rangle =0 \label{NM1}
 \end{align} 
  for any $g (\neq e) \in G$,
$e$ is the unit element.
The condition \eqref{NM1} holds only when
$|\phi\rangle$ is written as
\begin{align}
|\phi\rangle= t \sum_{l=1}^k \sqrt{p_l} \mathbbm{1}_l\otimes (F_l^\dagger)^{-1} |\phi_l\rangle
\end{align}
with a normalizing constant
because of the following relation:
\begin{align}
&\Big(t \sum_{l=1}^l \sqrt{p_l}
\langle \phi_l|
( \mathbbm{1}_l\otimes (F_l^\dagger)^{-1} )^\dagger\Big)
f(g)|\psi\rangle \notag\\
=&
t \sum_{l=1}^l \sqrt{p_l}
\langle \phi_l|
( \mathbbm{1}_l\otimes (F_l^\dagger)^{-1} )^\dagger
f(g) \sum_{l=1}^l \sqrt{p_l} \mathbbm{1}_l\otimes F_l |\phi_l\rangle \notag\\
=&
t \sum_{l=1}^l \sqrt{p_l}
\langle \phi_l|
( \mathbbm{1}_l\otimes F_l^{-1}) 
\Big(f(g) \sum_{l'=1}^k \sqrt{p_{l'}} \mathbbm{1}_l\otimes F_{l'} |\phi_{l'}\rangle \Big)\notag\\
=&
t \sum_{l=1}^l p_l \langle \phi_l| ( \mathbbm{1}_l\otimes F_l^{-1}) 
f(g)\otimes F_{l} |\phi_{l}\rangle \notag\\
=&
t \sum_{l=1}^l p_l
\langle \phi_l| f(g)\otimes \mathbbm{1}_l|\phi_{l}\rangle =0
\end{align}
for $g (\neq e) \in G$.
The normalizing constant $t$ is given as
\begin{align}
t^{-1}= \sum_{l=1}^k p_l  \Tr  (F_l^\dagger)^{-1} F_l^{-1}.
\end{align}
Then, we have
\begin{align}
\sum_{g \in G} f(g)|\phi\rangle \langle \phi| f(g)^\dagger
=&
\sum_{l=1}^k |G| p_l \frac{1}{d_l^2} \mathbbm{1}_l \otimes (F_l^\dagger)^{-1} F_l^{-1} 
\notag\\
=&
\sum_{l=1}^k  \mathbbm{1}_l \otimes (F_l^\dagger)^{-1} F_l^{-1} .
\end{align}
Therefore, $c_{|\psi\rangle}$ is the inverse of the maximum (over $l=1,\cdots k$) of the maximum eigenvalues of $(F_l^\dagger)^{-1} F_l^{-1}$.

For simplicity, we consider the case when $G$ is a commutative group.
All irreducible representations are one-dimensional. 
Hence, a state $|\psi\rangle$ is written as
$\sum_{l=1}^k\sqrt{\frac{1}{k}} f_l|\phi_l\rangle$,
where $f_l$ is a complex number and $|\phi_l\rangle$ is a basis on the one-dimensional irreducible space ${\cal H}_l$.
Hence,  $c_{|\psi\rangle}$ is $\min_l |f_l|^{-2}$.
 
\if0
In this case, the implementation domain $\ID$ satisfies $f(g)\ketbra{\psi}{\psi}f(g)^\dag\in \ID$ for any $g\in G$,
and the corresponding N-POVM measurement $\{N_g\}_{g\in G}$ is given as
\begin{align}
	N_g=\frac{1}{c_{\ket{\phi}}-1}\left(c_{\ket{\phi}} f(g)\ketbra{\phi}{\phi}f(g)^\dag +\frac{1}{|G|}\left(\mathbbm{1}-c_{\ket{\phi}} M_0\right)\right).
\end{align}

Especially,
when $G$ is a finite cyclic group with a primitive $g_0$,	
the corresponding N-POVM measurement is given as $\{N_i\}_{i=1}^{|G|}$ defined as
\begin{align}
	N_i=\frac{1}{c_{\ket{\phi}}-1}\left(c_{\ket{\phi}} f(g_0)^i\ketbra{\phi}{\phi}f(g_0)^{i\dag} +\frac{1}{|G|}\left(\mathbbm{1}-c_{\ket{\phi}} M_0\right)\right).
\end{align}
\fi

\section{Conclusion}

We have investigated the relationship between N-POVM measurements and post-selected POVM measurements.
We have given a constructive way to implement an N-POVM measurement by post-selected POVM measurement with restricting the domain (Theorem~\ref{theorem:general}).
Besides, in some situation, we have given a sufficient condition for a large dimensional domain and a large accepting probability (Theorem~\ref{theorem:dimension}),
which implies the estimation of the cost of the implementation of N-POVM.
Moreover, as a special case of our implementation, we have given a constructive way to regard post-selection as an N-POVM measurement with restricting the domain (Theorem~\ref{2theorem:post-selection}).

Our results imply a new relationship between N-POVM measurements in GPTs and post-selection in quantum theory.
Due to our results, N-POVM measurements in GPTs can be physically implemented in quantum theory by post-selection and restriction of the domain as a pre-information.
This is a new direction of applying studies of GPTs to quantum information theory.

Finally, we give the following future directions and open problems from our results.
\begin{enumerate}
	\item[(1)] Post-Selected Shannon Theory\\
	Recently, post-selection is well-studied in quantum information theory \cite{Aaronson2005,CF2015,Arvidsson2020,Regula2022,RLW2022}, but its accepting probability is not easy to estimate.
	Our new relationship can estimate the accepting probability of some post-selecting process by combining a bound performance of N-POVM in GPTs.

	\item[(2)] Upper Bound of the Performance in GPTs
	In GPTs, the bound performance for some information tasks is still open.
	Our results imply that the beyond-quantum performance of N-POVM measurement in GPTs is regarded as the advantage of post-selection.
	Therefore, our results can show a new bound performance of N-POVM measurement for some information tasks by applying the bound performance of post-selection in quantum theory.

	\item[(3)] Relation of Complexity\\
	In GPTs, a generalization of quantum computation is studied and its complexity is given by AWPP \cite{Lee2015,Barrett2019}.
	Similar results are known in quantum theory, that is, quantum computation with a low accepting rate of post-selection is also given by AWPP \cite{MN2016}.
	Can we show the mathematical reason of this mysterious relationship by applying our new relationship.

	\item[(4)] Relation of Weak-Value Measurement
	It is an open problem to give an physical interpretation of weak-value.
	A weak-value measurement is defined by a post-selected POVM measurement \cite{AAV1988}, and it is regarded as an N-POVM measurement by our results.
	This will imply a new physical interpretation of weak-values.
\end{enumerate}

\acknowledgments
HA appreciates Bartosz Regula and Yui Kuramochi for discussing applications of our results.

\newpage

\appendix

\section{Appendix}

\subsection{Proof of Lemma~\ref{lemma:domain}}
\label{append:lem}

\begin{proof}
Given two elements $\rho_1,\rho_2 \in \ID$,
we choose an index $i\in I$ such that $\Tr N_i \rho_1\neq \Tr N_i \rho_2$. 
Then, we have
\begin{align}
	&\cfrac{p\Tr \rho_1M_i}{\sum_{j\in I}\Tr \rho_1M_j}
	+\cfrac{(1-p)\Tr \rho_2M_i}{\sum_{j\in I}\Tr \rho_2M_j}\nonumber\\
	=&p\Tr \rho_1N_i+(1-p)\Tr \rho_2N_i\notag\\
	=&\Tr (p\rho_1+(1-p)\rho_2)N_i\notag\\
	=&\cfrac{\Tr (p\rho_1+(1-p)\rho_2) M_i}{\sum_{j\in I}\Tr (p\rho_1+(1-p)\rho_2) M_j}.\label{eq:lemma:domain-1}
\end{align}
For simplicity, we take values $a,b,x,y$ as
\begin{gather}
	a=\Tr \rho_1M_i, \quad b=\Tr \rho_2M_i,\\
	x=\sum_{j\in I}\Tr \rho_1M_j, \quad y=\sum_{j\in I}\Tr \rho_2 M_j.
\end{gather}
Then, the equation \eqref{eq:lemma:domain-1} is rewritten as
\begin{align}
	\frac{pa}{x}+\frac{(1-p)b}{y}=\frac{pa+(1-p)b}{px+(1-p)y}.
\end{align}
Because $x,y>0$, we reduce the above equation as follows:
\begin{align}
	0=
	&y(px+(1-p)y)pa+x(px+(1-p)y)(1-p)b\\\nonumber
	&-xy(pa+(1-p)b)\\
	=&ypa(px+(1-p)y-x)\\\nonumber
	&+x(1-p)b(px+(1-p)y-y)\\
	=&ypa(1-p)(y-x)+x(1-p)bp(x-y)\\
	=&(x-y)p(1-p)(bx-ay).
\end{align}
Because of the assumption $\Tr N_i \rho_1\neq \Tr N_i \rho_2$, we also obtain
\begin{align}
	0\neq
	&\Tr N_i \rho_1-\Tr N_i \rho_2\\
	=&\frac{\Tr \rho_1 M_i}{\sum_{j\in I}\Tr \rho_1M_j}-\frac{\Tr \rho_2 M_i}{\sum_{j\in I}\Tr \rho_2M_j}\\
	=&\frac{a}{x}-\frac{b}{y}=ay-bx.
\end{align}
Therefore, we conclude $x=y$, i.e.,
$\sum_{j\in I}\Tr \rho_1M_j
=\sum_{j\in I}\Tr \rho_1M_j$.
\end{proof}

\subsection{Proof of Theorem~\ref{theorem:general}}\label{append-1}

\begin{proof}
	\textbf{Outline}
	
	We need to show the following two statements:
	\textbf{(i)} The family $\bm{M}$ defined by \eqref{def:Im-POVM} is a POVM.
	\textbf{(ii)} The equation \eqref{eq:def-IWP} holds for any $\rho\in\cD(\{f_i^{(k)}\}_{i,k})$ and any $i\in I$.
	
	\noindent
	$\blacktriangleright$ \textbf{Proof of (i)}
	
	First, because of the choice of $\bm{M}$ in \eqref{def:Im-POVM},
	the summation condition $\displaystyle\sum_{j\in I\cup\{i_0\}}N_j=\mathbbm{1}$ holds.
	Next, because $S_i^{(k)}$ is positive semi-definite,
	the matrix $\displaystyle \sum_{i\in I}\sum_{k=1}^n S_i^{(k)}$ is also positive semi-definite,
	and therefore, $c>0$ holds.
	Hence, the matrix $M_j$ is positive semi-definite for $j\in I$.
	Because of the constant $c$ is the maximum eigenvalue of $\displaystyle \sum_{i\in I}\sum_{k=1}^n S_i^{(k)}$,
	the following inequality holds:
	\begin{align}
		\frac{1}{c}\sum_{i\in I}\sum_{k=1}^n S_i^{(k)} \le \mathbbm{1},
	\end{align}
	which implies the following inequality:
	\begin{align}
		\mathbbm{1}-\frac{1}{c}\sum_{i\in I}\sum_{k=1}^n S_i^{(k)}\ge0.
	\end{align}
	As a result, the family $\{M_j\}$ defined by \eqref{def:Im-POVM} is a POVM.
	
	\noindent
	$\blacktriangleright$ \textbf{Proof of (ii)}
	
	Because of the choice of the domain $\cD(\{f_i^{(k)}\}_{i,k})$ in \eqref{def:Im-domain},
	the following equation holds for any $\rho\in\cD(\{f_i^{(k)}\}_{i,k})$, any $i\in I$, and any $1\le k\le n$:
	\begin{align}\label{eq:proof-theorem1-1}
		\Tr \rho f_i^{(k)}(S_i^{(k)})=\Tr f_i^{(k)\dag}(\rho)S_i^{(k)}=\Tr \rho S_i^{(k)}.
	\end{align}
	Therefore, the following equation holds for any $\rho\in\cD(\{f_i^{(k)}\}_{i,k})$:
	\begin{align}
		&\Tr \rho \sum_{i\in I}\sum_{k=1}^n S_i^{(k)}
		=\sum_{i\in I}\sum_{k=1}^n \Tr \rho S_i^{(k)}\\
		=&\sum_{i\in I}\sum_{k=1}^n \Tr \rho f_i^{(k)}(S_i^{(k)})
		=\Tr \rho \sum_{i\in I} \sum_{k=1}^n f_i^{(k)}(S_i^{(k)})\\
		=&\Tr \rho \sum_{i\in I} N_i\stackrel{(a)}{=}\Tr \rho \mathbbm{1}=1.\label{eq:proof-theorem1-2}
	\end{align}
	The equation $(a)$ is shown by the summation condition of the beyond-quantum measurement $\bm{N}$.
	Then, we obtain the desirable equation \eqref{eq:def-IWP} for any $\rho\in\cD(\{f_i^{(k)}\}_{i,k})$ and any $i\in I$as follows:
	\begin{align}
		&\cfrac{\Tr \rho M_i}{\sum_{j\in I}\Tr \rho M_j}
		=\frac{1}{c}\frac{\Tr \rho \sum_{k=1}^n S_i^{(k)}}{\sum_{j\in I}\Tr \rho \frac{1}{c}\sum_{k=1}^n S_i^{(k)} }\\
		=&\frac{\Tr \rho \sum_{k=1}^n S_i^{(k)}}{\Tr \rho \sum_{j\in I}\sum_{k=1}^n S_i^{(k)} }
		\stackrel{(a)}{=}\Tr \rho \sum_{k=1}^n S_i^{(k)}\\
		\stackrel{(b)}{=}&\Tr \rho \sum_{k=1}^n f_i^{(k)}\left(S_i^{(k)}\right)=\Tr \rho N_i.
	\end{align}
	The equation $(a)$ is shown from the equation \eqref{eq:proof-theorem1-2}.
	The equation $(b)$ is shown from the equation \eqref{eq:proof-theorem1-1}.
\end{proof}

\subsection{Generalization of Theorem~\ref{theorem:general}}
\label{append-extend}

In this section, we give a generalization of Theorem~\ref{theorem:general} from quantum implementation to more general implementations.
In order to state the generalization,
we introduce a model of GPTs in detail.

Let $\cV$ be a finite-dimensional real vector space,
and we consider a model determined by a state space $\cS\subset\cV$.
As a standard assumption of GPTs \cite[Section 2]{takakura2022}, the state space $\cS$ spans $\cV$.
A measurement $\{M_i\}_{i\in I}$ is defined as a family of $\cV^\ast$ such that $M_i(\rho)\ge0$ for any $\rho\in\cS$ and $\sum_{i\in I} M_i=\mathbbm{1}$,
where $\mathbbm{1}$ is the functional taking 1 for all inputs.
An index $i$ of a measurement corresponds to its outcome.
Here, we denote $\cE_\cS$ and $\cM_{\cS}$ as the set of all functional satisfying $M_i(\rho)\ge0$ for any $\rho\in\cS$.
Also, we denote $\cM_{\cS}$ measurements associated with $\cS$.
The set $\cE_\cS$ also spans $\cV^\ast$ \cite[Section 2]{takakura2022}.

Also, we define a measurement beyond $\cM_\cS$ as a family $\{N_i\}_{i\in I}$ such that $\sum_{i\in I} M_i=\mathbbm{1}$ but $M_i(\rho)\ge0$ does not hold for a state $\rho\in\cS$.
Similarly to the main topic,
even if a family $\{N_i\}_{i\in I}$ is beyond $\cM_\cS$,
it can be implemented by a measurement $\{M_i\}_{i \in I\cup\{i_0\}}$ with post-selection and restriction of the domain.
We say that $\{N_i\}_{i\in I}$ is implemented by $\{M_i\}_{i \in I\cup\{i_0\}}$ with post-selection in an implementation domain $\ID$
when the following equation holds for any $\rho\in\ID$:
\begin{align}\label{eq:ps-ext}
	N_i(\rho)=\frac{M_i(\rho)}{\sum_{j\in I} M_j(\rho)}.
\end{align}

Similarly to Theorem~\ref{theorem:general},
we prove the following theorem.
\begin{theorem}\label{theorem:extend}
	Let $\bm{N}:=\{N_i\}_{i\in I}$ be a measurement beyond $\cM_\cS$ written as
	\begin{align}\label{eq:representation-ext}
		N_i=\sum_{k=1}^n f_i^{(k)}(S_i^{(k)}),
	\end{align}
	where $f_i^{(k)}$ and $S_i^{(k)}$ are a linear map on $\cV^\ast$ and a functional in $\cE_\cS$, respectively.
	Define a set $\cD(\{f_i^{(k)}\}_{i,k})$ as
	\begin{align}\label{def:Im-domain-ext}
	\begin{aligned}
		\cD(\{f_i^{(k)}\}_{i,k}):=&\{\rho\in\cS,\\
		&f_i^{(k)\dag}(\rho)=\rho \ ( i\in I, \ 1\le k\le n)\}.
	\end{aligned}
	\end{align}
	Then, $\bm{N}$ is implemented by the following measurement $\bm{M}(\{f_i^{(k)},S_i^{(k)}\}):=\{M_j\}_{j\in I\cup\{i_0\}}$ with post-selection on the implementation domain $\ID=\cD(\{f_i^{(k)}\}_{i,k})$:
	\begin{align}
		M_j=
		\begin{cases}
			\displaystyle \frac{1}{\overline{c}({\{S_i^{(k)}\}})}\sum_{k=1}^n S_i^{(k)} & (j=i\in I),\\
			\displaystyle \mathbbm{1}-\frac{1}{\overline{c}({\{S_i^{(k)}\}})}\sum_{i\in I}\sum_{k=1}^n S_i^{(k)} & (j=i_0),
		\end{cases}\label{def:Im-ext}
	\end{align}
	where $\overline{c}({\{S_i^{(k)}\}})$ is the maximum value of $\displaystyle \sum_{i\in I}\sum_{k=1}^n S_i^{(k)}(\rho)$ for $\rho\in\cS$.
\end{theorem}

\begin{proof}[Proof of Theorem~\ref{theorem:extend}]
	We need to show the following two statements:
	\textbf{(i)} The family ${M_i}_{i\in I\cup \{i_0\}}$ defined by \eqref{def:Im-ext} is a measurement in $\cM_\cS$.
	\textbf{(ii)} The equation \eqref{eq:ps-ext} holds for any $\rho\in\cD(\{f_i^{(k)}\}_{i,k})$ and any $i\in I$.

	\noindent
	$\blacktriangleright$ \textbf{Proof of (i)}

	The condition $\sum_{i\in I\cup\{i_0\}}M_i=\mathbbm{1}$ trivially holds because of the choice of $M_i$.
	Also, the condition $M_i(\rho)\ge0$ trivially holds for $\rho\in\cS$ and $i\in I$ because of the choice of $M_i$.
	In the case of $i=i_0$,
	any $\rho\in\cS$ satisfies $M_i(\rho)\ge0$
	because the definition $\overline{c}({\{S_i^{(k)}\}})$ implies the inequality
	\begin{align}
		\frac{1}{\overline{c}({\{S_i^{(k)}\}})}\sum_{i\in I}\sum_{k=1}^n S_i^{(k)}\le 1
	\end{align}
	for any $\rho\in\cS$.
	As a result, the family $\{M_i\}_{i I\cup\{i_0\}}$ is a measurement in $\cM_\cS$.

	\noindent
	$\blacktriangleright$ \textbf{Proof of (ii)}

	Because of the choice of the domain $\cD(\{f_i^{(k)}\}_{i,k})$ in \eqref{def:Im-domain-ext},
	the following equation holds for any $\rho\in\cD(\{f_i^{(k)}\}_{i,k})$, any $i\in I$, and any $1\le k\le n$:
	\begin{align}\label{eq:proof-theorem1-1}
		f_i^{(k)}(S_i^{(k)})(\rho)= (S_i^{(k)})(f_i^{(k)\dag}(\rho))=S_i^{(k)}(\rho).
	\end{align}
	Therefore, similarly to the proof of Theorem~\ref{theorem:general},
	the equation \eqref{eq:ps-ext} holds for any $\rho\in\cD(\{f_i^{(k)}\}_{i,k})$ and any $i\in I$.
\end{proof}

\subsection{Proof of Theorem~\ref{theorem:ppt}}
\label{append-3}
\begin{proof}
	Because $\Gamma$ does not change the trace, we obtain the following equality:
	\begin{align}
		&\Tr \sum_i \left(S_i^{(1)}+S_i^{(2)}\right)
		= \Tr \sum_i \left(S_i^{(1)}+\Gamma(S_i^{(2)})\right)\\
		=& \Tr \sum_i N_i=\Tr I=d.
	\end{align}
	The matrix $S_i^{(k)}$ is positive semi-definite, and therefore, the largest eigenvalue of $ \sum_i \left(S_i^{(1)}+S_i^{(2)}\right)$ is less than its trace.
	As a result, we obtain $\mathrm{Acc}(\bm{M})\ge 1/d$.
\end{proof}

\subsection{Proof of Theorem~\ref{theorem:dimension}}
\label{append-2}

Because $\{\rho_i\}_{i=1}^d$ is the family of orthogonal pure states on $d$-dimensional Hilbert space,
$\{\rho_i\}_{i=1}^d$ determines an orthonormal basis.
As a preliminary,
we represent matrices by considering the basis determined by $\{\rho_i\}_{i=1}^d$ in this subsection.

In order to prove Theorem~\ref{theorem:dimension},
we give the following lemma.

\begin{lemma}\label{lemma:dimension}
	If an N-POVM $\bm{N}=\{N_i\}_{i\in I}$ satisfies the conditions C1 and C2 in Theorem~\ref{theorem:dimension},
	the diagonal entries of $N_i$ is non-negative and strictly positive at least for one entry under the representation of the basis determined by $\{\rho_i\}_{i=1}^d$.
\end{lemma}

\begin{proof}[Proof of Lemma~\ref{lemma:dimension}]
	Because of the condition C1, i.e., $\{\rho_i\}_{i=1}^d\in\QD{\bm{N}}$,
	$\Tr \rho_jN_i\ge0$ holds,
	and therefore,
	the diagonal entries of $N_i$ is non-negative.
	Next, if all diagonal entries of $N_i$ are 0 but $N_i$ is non-zero,
	then, there exists an non-zero off-diagonal entry $n^{(i)}_{lk}$.
	However, because $\cB_{\epsilon}(\mathbbm{1}/d)$ is contained in $\QD{\bm{N}}$,
	the normalization $\overline{X}_{\pm}$ of the matrix $X_\pm:=I/d\pm(\sqrt{\epsilon}/2) (\ketbra{l}{k}+\ketbra{k}{l})$ belongs to $\QD{\bm{N}}$.
	Then, either relation $\Tr X_+N_i<0$ or $\Tr X_-N_i<0$ holds, which contradict to $\overline{X}_\pm\in\QD{\bm{N}}$.
\end{proof}

\begin{proof}[Proof of Theorem~\ref{theorem:dimension}]
	\textbf{Outline}
	
	We choose a desirable family of linear functions $\{f_i\}$ and positive semi-definite matrices $\{S_i\}$ by forcusing an orthogonal basis determined by $\{\rho_i\}_{i=1}^d$ (Part (i)).
	By the choise of $\{f_i\}$ we show the desirable properties \eqref{eq:dim-2} and \eqref{eq:dim-3} (Part (ii)).
	
	\noindent
	$\blacktriangleright$ \textbf{Part (i)}

	Let $E_j \ (j=1,\cdots, d)$ be the diagonal matrix $E_j=\ketbra{j}{j}$.
	Now, we define sets of indices
	\begin{align}
		J_1:&=\{1,\cdots, d\},\\
		J_2:&=\{d+1,\cdots,\dim'(\bm{N})\},\\
		J_3:&=\{\dim'(\bm{N})+1,\cdots,d^2\}.
	\end{align}
	Because $\Her{\cH}$ is $d^2$-dimensional vector space,
	there exists an orthonormal basis $\{B_j\}_{j=1}^d$ such that
	\begin{align}
		B_j&=E_j \ (j\in J_1),\\
		B_j&\in \spn{\{N_i\}_{i\in I'}}\cap\spn{\{E_{j'}\}_{j'\in J_1}}^\perp \ (j\in J_2),\\
		B_j&\in \spn{\{N_i\}_{i\in I'}}^\perp\cap\spn{\{E_{j'}\}_{j'\in J_1}}^\perp \ (j\in J_3).
	\end{align}

	We denote $\hat{N}_i$ and $\tilde{N}_i$ as the diagonal and off-diagonal parts of $N_i$, respectively.
	Lemma~\ref{lemma:dimension} ensures $\hat{N}_i\neq 0$,
	and therefore, there exists a traceless Hermitian matrix $\overline{N}_i$ such that $\hat{N}_i+\overline{N}_i\ge0$.
	Then, we choose a linear function $f_i$ for $i\in I'$ such that
	\begin{align}
		f_i(B_j)&=B_j \ (j\in J_1\cup J_3),\\
		f_i(\overline{N}_i)&=\tilde{N}_i.
	\end{align}
	The function $f_i$ always exists because $\Tr B_j\overline{N}_i=0$ for any $i,j$.
	Also, $f_i$ satisfies
	\begin{align}
		f_i(\hat{N}_i+\overline{N}_i)=\hat{N}_i+\tilde{N}_i=N_i,
	\end{align}
	and $S_i:=\hat{N}_i+\overline{N}_i\ge0$.
	As a result, we have chosen $\{f_i\}$ and $\{S_i\}$ satisfying the condition in Theorem~\ref{theorem:general},
	and we choose a POVM $\bm{M}$ implemented $\bm{N}$ by Theorem~\ref{theorem:general}.

	\noindent
	$\blacktriangleright$ \textbf{Part (ii)}

	First, because $\cB_{\epsilon}(\mathbbm{1}/d)\subset\QD{\bm{N}}$,
	the relation $\dim(\QD{\bm{N}})=d^2$ holds.
	Therefore, the implementation domain $\cD(\{f_i\})$ given in \eqref{def:Im-domain} satisfies the relation \eqref{eq:dim-2} as follows:
	\begin{align}
		&\dim(\cD(\{f_i\}))\nonumber\\
		=&\dim(\spn{\{B_i\}_{i=1}^d})\nonumber\\
		+&\dim\left(\spn{\{\check{N}_i\}_{i\in I}}^{\perp} \cap \spn{\{\overline{N}_i\}_{i\in I}}^{\perp} \right)\\
		\ge &d+d^2-2\dim(\spn{\{N_i\}_{i\in I}}.
	\end{align}
	Second,
	the function $f_i$ satisfies $f_i(B_j)=B_j$ for $j=1,\cdots,d$ and maps from traceless matrices to traceless matrices,
	and therefore, $f_i$ is trace-preserving.
	As a result, Theorem~\ref{theorem:ppt} shows the relation \eqref{eq:dim-3}.
\end{proof}

\subsection{Proof of Theorem~\ref{2theorem:post-selection}}
\label{append-4}
\begin{proof}
(iii)$\Rightarrow$(ii) is trivial.
First, we will show (ii)$\Rightarrow$(i).
\eqref{MNB4} implies that
\begin{align}
\sum_{i \in I}\Tr \rho M_i=1-\frac{1}{c_0}=\frac{c_0-1}{c_0}.
\end{align}
Hence, 
\eqref{2eq:PS-2} implies that
\begin{align}
&\Tr \rho \mathbbm{1}=
\sum_{i \in I}\Tr \rho N_i
=\frac{c_0}{c_0-1}\sum_{i \in I}\Tr \rho M_i \notag\\
=&\frac{c_0}{c_0-1}\Tr \rho (\mathbbm{1}-M_{i_0})\label{BVC}
\end{align}
for $\rho \in \ID$. The relation \eqref{BVC} is equivalent to \eqref{MNB1}.
Hence, we obtain (i).

Next, we will show (i)$\Rightarrow$(iii).
%We define $N_i:= (1+c)M_i- M_{i_0}$.
%Assume $|I|=k_0$.
The following equation shows $\sum_{i\in I} N_i =\mathbbm{1}$:
\begin{align}
&\sum_{i\in I}N_i\nonumber\\
=&\sum_{i\in I}
\frac{c_0}{c_0-1}M_i-\frac{1}{(c_0-1)|I|}(\mathbbm{1}-c_0 M_{i_0})
\\
=&\frac{c_0}{c_0-1}(\mathbbm{1}- M_{i_0})-\frac{1}{c_0-1}(\mathbbm{1}-c_0 M_{i_0})\\
=&\mathbbm{1}+\left(-\frac{c_0}{c_0-1}+\frac{c_0}{c_0-1}\right)M_{i_0}\\
=&\mathbbm{1}.
\end{align}

Next, due to the relation \eqref{MNB1},
we obtain the following equation for any $\rho\in\ID$:
\begin{align}
	&\sum_{j\in I}\Tr M_j\rho\nonumber\\
	=&\Tr\left(\mathbbm{1}-M_{i_0}\right) \rho\\
	=&1-\Tr M_{i_0}\rho\\
	=&1-\frac{1}{c_0}=\frac{c_0-1}{c_0}.\label{eq:PS-proof}
\end{align}
As a result, we show the equation \eqref{2eq:PS-2} for $\rho \in \ID$ and $i\in I$ by \eqref{eq:PS-proof} as follows:
\begin{align}
&\Tr \rho N_i\nonumber\\
=&\Tr \rho
\left(\frac{c_0}{c_0-1}M_i-\frac{1}{(c_0-1)|I|}(\mathbbm{1}-c_0 M_{i_0})\right)
\\
=&\frac{c_0}{c_0-1}\Tr \rho M_i\\
=&\frac{\Tr M_i\rho}{\sum_{j\in I} \Tr M_j\rho}.
\end{align}
\end{proof}

\end{document}